\documentclass[twoside, 10pt]{article}
\usepackage{amsmath,amsthm,amssymb,amscd,amstext,amsfonts}
\usepackage[utf8]{inputenc}
\usepackage{mathtools}
\usepackage{mathrsfs}
\usepackage{thmtools}
\usepackage{latexsym}
\usepackage{verbatim}
\usepackage{framed}
\usepackage{graphicx}
\usepackage{stmaryrd}
\usepackage{enumitem}
\usepackage{fullpage}
\usepackage{bm}
\usepackage{url}
\usepackage{physics}
\usepackage{dsfont}
\usepackage{todonotes}
\usepackage{tikz}
\usepackage{graphicx}
\usepackage{titling}
\usepackage{color}
\usepackage{xpatch}

\usepackage[small, explicit]{titlesec}
\titlelabel{\thetitle.\enspace}

\usepackage[linktocpage=true,]{hyperref}

\hypersetup{
	colorlinks = true,
    linkcolor={purple},
    urlcolor={purple},
    citecolor={blue!80!black}
}

\makeatletter
\def\tagform@#1{\maketag@@@{(\ignorespaces{\oldstylenums{#1}}\unskip\@@italiccorr)}}
\renewcommand{\eqref}[1]{\textup{{\normalfont(\oldstylenums{\ref{#1}}}\normalfont)}}
\let\over\@@over
\let\atop\@@atop
\makeatother

\titleformat{\paragraph}[runin]%
{\normalfont\normalsize\bfseries}{\theparagraph}{1em}{#1}[.]
\titlespacing{\paragraph}{0pt}{1.5ex plus 1ex minus .2ex}{0.7em}

\usepackage[nameinlink, noabbrev, capitalize]{cleveref}

\usepackage[letterpaper, total={5.6in, 8in}, vmarginratio=1:1, hmarginratio=1:1, headsep=21pt]{geometry}

\declaretheoremstyle[bodyfont=\normalfont\slshape, notefont=\normalfont\itshape,notebraces={{\rm(}}{{\rm)}}, postheadspace=0.5em,headpunct={\rm.}, spaceabove=8pt, spacebelow=8pt]{slbody}

\declaretheorem[name=Theorem, numberwithin=section, style=slbody]{theorem}
\declaretheorem[name=Lemma, numberwithin=section, sibling=theorem, style=slbody]{lemma}
\declaretheorem[name=Corollary, numberwithin=section, sibling=theorem, style=slbody]{corollary}
\declaretheorem[name=Proposition, numberwithin=section, sibling=theorem, style=slbody]{proposition}

\newcommand\normrow[1]{\left|\!\left|#1\right|\!\right|_{\rm row}}
\newcommand\normcol[1]{\left|\!\left|#1\right|\!\right|_{\rm col}}
\newcommand\normgamma[1]{\left|\!\left|#1\right|\!\right|_{\gamma_2}}

\newcommand\textmc[1]{{\small#1}}

\newcommand{\RR}{\mathbf{R}}   
\newcommand{\CC}{\mathbf{C}}   
\newcommand{\QQ}{\mathbf{Q}}   
\newcommand{\ZZ}{\mathbf{Z}}   

\renewcommand{\O}{\mathcal{O}}

\newcommand{\EQ}{\hbox{\rm\scriptsize EQ}}
\newcommand{\PL}{\hbox{\rm\small PL}}
\DeclareMathOperator\D{D}  
\DeclareMathOperator\R{R} 

\DeclareMathOperator\sgnrk{sgnrk} 
\DeclareMathOperator\sgn{sgn} 

\newcommand{\one}{\mathop{\mathbf{1}}\nolimits}

\newcommand{\ex}{\operatorname{E}}

\renewcommand{\maketitle}{%
  \begin{center}
    {\large\bf Communication complexity of point-line incidences\\ \smallskip over the reals }\\
    \smallskip
    {}
    \vskip 30pt
    {\sc Marcel K.\ Goh\enspace{\rm and}\enspace Hamed Hatami}
    \medskip

    \vskip 30pt
  \end{center}
}

\title{}
\author{}
\date{}

\begin{document}

\maketitle

\renewenvironment{abstract}{\quotation\noindent\small{\bfseries\abstractname.}}{\endquotation}

\begin{abstract}
We construct a point-line incidence problem over the reals whose randomized communication complexity is constant, but whose deterministic communication complexity is linear even when the players have access to an equality oracle. 
This is the strongest possible separation between these two measures, and it improves on an earlier $O(1)$-versus-$\Omega(\sqrt{n})$ separation of G\"o\"os, Harms, and Riazanov.

Because point-line incidence problems have constant sign rank, our construction also bears on a question of Harms and Zamaraev, namely whether constant sign rank together with constant randomized communication complexity forces constant equality-oracle complexity. This was already refuted by G\"o\"os, Harms, Imbach, and Sokolov with a logarithmic lower bound; our example improves the separation to linear, which is optimal.

The proof draws on a construction in the recent disproof of the sum-product conjecture over the reals by Bloom, Sawin, Schildkraut, and Zhelezov, using totally real number fields of large degree and small discriminant.

\end{abstract}

\vskip50pt
\baselineskip=13.5pt

\section{Introduction}

Although randomized communication can be significantly more powerful than deterministic communication, there are remarkably few natural problems exhibiting such a gap. The classical example is verifying the equality of two $n$-bit strings:~this problem has randomized communication complexity $O(1)$ via hashing, but deterministic communication complexity $n+1$, the largest possible value for any $n$-bit communication problem.

In fact, most known problems with constant randomized communication complexity are built, in a precise sense, out of equality itself: they have deterministic protocols of small cost once the players are allowed to make equality queries.

This raises a natural question: \textsl{Is the power of randomness in communication, at least for natural low-complexity problems, essentially equivalent to the power to solve equality?}

In this paper, we give an optimal negative answer using a basic algebraic identity-testing problem. We exhibit an arrangement of points and lines in $\RR^2$ for which the point-line incidence problem has constant randomized communication complexity, but requires linear communication in the deterministic equality-oracle model.

\paragraph{Communication with equality queries} An $n$-bit communication problem is a Boolean matrix 
$F:\{0,1\}^n\times \{0,1\}^n \to \{0,1\}$. Alice receives $u \in \{0,1\}^n$, Bob receives $v \in \{0,1\}^n$, and their goal is to compute $F(u,v)$ while exchanging as few bits as possible.

Let $\D(F)$ denote the deterministic communication complexity of $F$, and let $\R_\epsilon(F)$ denote its \emph{public-coin} randomized communication complexity with worst-case error at most $\epsilon$. We write $\R(F)=\R_{1/3}(F)$. The choice of $1/3$ is not essential, since the error can be reduced by independent repetition of the protocol and taking the majority vote of the outcomes at only a constant-factor increase in communication.

We also consider deterministic communication with access to an equality oracle. Instead of communicating bits over a channel, Alice and Bob both have access to an oracle that takes an arbitrary bitstring $s$ from Alice and an arbitrary bitstring $t$ from Bob and broadcasts the bit $\one_{[s=t]}$ back to both of them. Each such query has unit cost. The strings $s$ and $t$ may depend arbitrarily on the players' inputs and on the previous transcript. The minimum number of oracle queries required to compute $F$ is denoted by $\D^{\EQ}(F)$. 
 
Since equality itself has randomized communication complexity $O(1)$, every equality query can be simulated by a constant-cost randomized protocol. Hence, a standard error-reduction implies that 
\[ 
\R(F) \;\lesssim\; \D^{\EQ}(F) \log \D^{\EQ}(F),
\]
where here and throughout the paper, we  
write $f\lesssim g$ and $f\gtrsim g$ for $f = O(g)$ and $f = \Omega(g)$ respectively. When the implied constants are not absolute, depending, say, on some parameter $k$, then we write, e.g., $f = O_k(g)$.

\subsection{Point-line incidences} In this paper we study the point-line incidence problem. Let $U$ be a finite set of lines in $\RR^2$, and let $V$ be a finite set of points in $\RR^2$. Alice receives a line $
u=(a,b)\in U$, representing $y=ax+b$, and Bob receives a point
$v=(x,y)\in V$. Their goal is to decide whether the point lies on the line, or equivalently whether
\begin{equation}
\label{eq:point-line}
ax+b-y=0.
\end{equation}
We denote this communication problem by $\PL_{U,V}$.

The point-line incidence problem is  the simplest nontrivial algebraic identity-testing problem beyond equality as fixing any of the variables to a constant reduces the problem to an instance of the equality problem. Our main result shows that there exists an arrangement of points and lines for which the corresponding incidence problem exhibits the strongest possible separation between randomized communication and equality-oracle deterministic communication complexities.
 
\begin{theorem}
\label{thm:main}
There exist sets $U, V \subseteq \RR^2$ of arbitrarily large size $|U| = |V| = 2^n$ such that
\[
\R(\PL_{U,V}) = O(1)
\qquad\text{and}\qquad
\D^{\EQ}(\PL_{U,V}) = \Omega(n).
\]
\end{theorem}

A few remarks are in order. Previously, the best known separation between randomized communication complexity and equality-oracle deterministic communication complexity was $O(1)$-versus-$\Omega(\sqrt{n})$, due to G\"o\"os, Harms, and Riazanov~\cite{GoosHarmsRiazanov2025}. Theorem~\ref{thm:main} improves this to the optimal bound $O(1)$-versus-$\Omega(n)$. The existence of such a separation was posed explicitly as Open Problem~5 in~\cite{CheungHatamiHosseiniNikolovPitassiShirley2026}.

Constant-degree algebraic identities such as $y = ax + b$ have bounded \emph{sign rank}. It was shown in~\cite{MR4494342} that any problem $F$ with $\D^{\EQ}(F) = O(1)$ must satisfy both $\R(F) = O(1)$ and $\sgnrk(F) = O(1)$. Harms and Zamaraev~\cite{HarmsZamaraev2024} asked whether the converse holds: does bounded sign rank and bounded randomized communication complexity imply bounded $\D^{\EQ}$? This was recently refuted by~\cite{GoosHarmsImbachSokolov2025}, who exhibited a problem with constant sign rank and constant randomized complexity, yet $\D^{\EQ} = \Omega(\log n)$. Theorem~\ref{thm:main} gives a much stronger negative answer to this question: our result shows that $\D^{\EQ}$ can in fact be \emph{linear} for problems with both constant sign rank and constant randomized communication complexity.

As we show later in Theorem~\ref{thm:main2}, the point-line incidence problem of Theorem~\ref{thm:main} satisfies
\begin{equation}
\label{eq:factorization_separation}
\widetilde\gamma_2(\PL_{U,V}) = O_\delta(1)
\qquad\text{and}\qquad
\normgamma{\PL_{U,V}} \ge N^{1/6-\delta},
\end{equation}
where $N = |U| = |V|$, the constant $\delta$ can be taken arbitrarily small,  and $\widetilde\gamma_2(\PL_{U,V})$ denotes the approximate $\gamma_2$-norm of $\PL_{U,V}$, i.e., the smallest value of $\normgamma{A}$ over all real matrices $A$ that approximate $\PL_{U,V}$ entrywise to within additive error $1/3$. (A precise definition of $\normgamma A$ is deferred to Section~3, where we discuss the factorization norm in earnest.) The question of separating the factorization norm from randomized communication complexity and the approximate $\gamma_2$-norm dates back to a 2009 paper of Linial and Shraibman~\cite{LinialShraibman2009}.
 
\subsection{Technical overview}
\label{sec:overview} 

Our goal is a point-line arrangement with $\D^{\EQ}(\PL_{U,V})$ large but $\R(\PL_{U,V})$ constant. The first half is, by now, essentially understood:~as we shall explain in Section 3, a linear lower bound on $\D^{\EQ}$ follows from producing an arrangement of $N$ points and $N$ lines with $N^{1+\Omega(1)}$ incidences.
The difficulty lies entirely in keeping the number of incidences this large while simultaneously forcing the randomized communication complexity down to a constant.

\paragraph{The integer grid}
Consider the integer grid
$$\bigl[-\sqrt{m}, \sqrt{m}\bigr] \times [-m, m],$$
arguably the simplest point-line arrangement. Here $N = \Theta(m^{3/2})$ and the number of incidences is $\Theta(m^2) = \Theta(N^{4/3})$, already giving the desired lower bound $\D^{\EQ}(\PL_{U,V}) \gtrsim n$ (where $n = \log N$).

The randomized complexity, however, is where this particular arrangement fails. The integer grid admits a simple protocol of cost $O(\log n)$. Alice and Bob jointly sample a random prime $p \lesssim n^{O(1)}$ and verify the identity $ax + b = y$ modulo $p$, exchanging only the residues of their inputs mod $p$. This already gives a separation $\R(\PL_{U,V}) \lesssim \log n$ versus $\D^{\EQ}(\PL_{U,V}) \gtrsim n$. One might hope to improve the protocol to constant cost, but in a beautiful recent paper, G\"o\"os, Harms, Richter, and Sofronova~\cite{GoosHarmsRichterSofronova2026} gave an elegant application of the circle method to show that $\R(\PL_{U,V}) = \Theta(\log n)$. In other words, the protocol of taking residues modulo $p$ is essentially optimal.

Our contribution is a point-line arrangement built from the algebraic integers of a carefully chosen number field, that retains $N^{1+\Omega(1)}$ incidences yet admits a constant-cost randomized protocol.

\paragraph{Algebraic integers}
It seems implausible at first glance that an arrangement with as many incidences as the integer grid could have a much lower randomized communication complexity. We construct such an arrangement by first defining $U$ and $V$ in $K^2$, for a suitable totally real number field $K$, and then embedding the construction into $\RR^2$. This approach is inspired by the recent disproof of the sum-product conjecture over the reals by Bloom, Sawin, Schildkraut, and Zhelezov~\cite{BloomSawinSchildkrautZhelezov2026}.

The general idea is as follows. We work in a totally real number field $K$ of degree $d$. Concretely, $K$ is a subfield of $\RR$ that is a $d$-dimensional vector space over $\QQ$. It is a fundamental result of Galois theory that there are exactly $d$ ways $\sigma_1,\ldots,\sigma_d$ to view the elements of $K$ as ordinary real numbers, where each $\sigma_i:K \to \RR$ is a field embedding. The coordinates of our points and lines will be \emph{algebraic integers}, the elements of a distinguished subring $\O_K \subseteq K$ that behaves like the integers inside $\QQ$.

The one fact we need is a quantitative form of the statement that nonzero algebraic integers cannot be too small: for any nonzero $z \in \O_K$, the values $\sigma_1(z), \ldots, \sigma_d(z) \in \RR$ cannot all be close to $0$. More precisely, averaging over a uniformly random index $i \in [d]$, a nonzero $z \in \O_K$ satisfies $|\sigma_i(z)| = \Omega(1)$ in expectation.  

Now fix a large constant $C = O(1)$ and restrict the coordinates of our points and lines to the bounded set
\[
B(K,C) \;=\; \bigl\{x \in \O_K : |\sigma_i(x)| \le C \text{ for all } 1 \le i \le d\bigr\}.
\]
Every fixed embedding $\sigma_i$ maps $B(K,C)$ into the small interval $[-C,C] \subseteq \RR$, so if $B(K,C)$ is large, many of them are inevitably sent close together. However, as noted above, any two distinct elements are sent far apart under a randomly chosen embedding.

Given $a,b,x,y \in B(K,C)$, the line $u = (a,b)$ passes through the point $v = (x,y)$ precisely when the algebraic integer $z = ax + b - y$ is zero. To detect this probabilistically, Alice and Bob jointly sample a random index $i \in [d]$ and estimate
\[
\sigma_i(z) \;=\; \sigma_i(a)\sigma_i(x) + \sigma_i(b) - \sigma_i(y)
\]
to sufficient accuracy, exchanging only $O(1)$ bits. If $z = 0$ then $\sigma_i(z) = 0$ for every $i$, whereas if $z \neq 0$ then, by the fact above, $|\sigma_i(z)|$ is bounded away from $0$ for a constant fraction of the indices $i$; for such an $i$, their estimate of $\sigma_i(z)$  is far enough from $0$ to certify $z \neq 0$.  This yields the desired constant-cost randomized protocol.

Finally, although for a general $K$ the set $B(K,C)$ can be small, a deep result of Martinet~\cite{Martinet1978} guarantees the existence of totally real number fields of arbitrarily large degree $d$ for which $|B(K,C)| = \Theta(C)^d$. As we will see, this is precisely the density needed for the resulting arrangement to have $N^{1+\Omega(1)}$ incidences, which in turn yields the $\D^{\EQ}(\PL_{U,V}) = \Omega(\log N)$ lower bound.

\subsection{Related work}

The equality-oracle deterministic communication complexity, and the question of its strength, was already considered in the seminal paper of Babai, Frankl, and Simon~\cite{BFS86}, which introduced communication complexity classes.

The first separation between $\R(\cdot)$ and $\D^{\EQ}(\cdot)$ was obtained by Chattopadhyay, Lovett, and Vinyals~\cite{ChattopadhyayLovettVinyals2020}, who gave an $O(\log n)$-versus-$\Omega(n)$ separation. Subsequently,~\cite{CheungHatamiHosseiniShirley2025,CheungHatamiHosseiniNikolovPitassiShirley2026} showed that the same bounds hold for the point-line incidence problem $\PL_{U,V}$, where $U,V = [-m,m]^2$. In the other direction, G\"o\"os, Harms, Richter, and Sofronova~\cite{GoosHarmsRichterSofronova2026} very recently showed that $\R(\PL_{U,V}) \gtrsim \log n$ in this case.

The first separations of the form $O(1)$-versus-$\omega(1)$ were given independently in~\cite{HambardzumyanHatamiHatami2023} and~\cite{MR4969131}, by considering the problem of determining whether the Hamming distance between two inputs $u,v \in \{0,1\}^n$ is at most some fixed $k = O(1)$. The randomized communication complexity of this problem is $O_k(1)$, while~\cite{HambardzumyanHatamiHatami2023} showed that its equality-oracle deterministic complexity is $\Theta_k(\log n)$.

The literature's inventory of constant-cost randomized communication problems is extremely meagre. A new construction was discovered in 2024 by Sherstov and Storozhenko~\cite{ShersStor2024}, and G\"o\"os, Harms, and Riazanov~\cite{GoosHarmsRiazanov2025} used it to prove the separation $\R(F) = O(1)$ versus $\D^{\EQ}(F) \gtrsim \sqrt{n}$. Prior to our work, this was the only known example of a constant-cost randomized communication problem with super-logarithmic equality-oracle complexity.

\paragraph{Sign rank}
The \emph{sign rank} of a Boolean matrix $A$, denoted $\sgnrk(A)$, is the minimum rank of a real matrix $B$ with $\sgn(B_{i,j}) = (-1)^{A_{i,j}}$ for all entries $i,j$. This notion was introduced in 1986 by Paturi and Simon~\cite{paturi1986probabilistic}, who showed that the optimal \emph{private-coin} communication complexity with error strictly less than $1/2$ (the \emph{unbounded-error communication complexity}) is $\log\sgnrk(A) \pm O(1)$.

Sign rank is a fundamental and well-studied notion in theoretical computer science and discrete geometry. Most natural geometric relations defined in a constant-dimensional space, such as point-line incidences, unit-distance graphs, and order relations among points, have constant sign rank, since membership in such a relation is typically decided by the sign of a fixed constant-degree polynomial in a bounded number of real coordinates.

As mentioned earlier, \cite{MR4494342} showed that any problem $F$ with $\D^{\EQ}(F) = O(1)$ must satisfy $\R(F) = O(1)$ and $\sgnrk(F) = O(1)$, and Harms and Zamaraev~\cite{HarmsZamaraev2024} asked whether the converse holds. G\"o\"os, Harms, Imbach, and Sokolov~\cite{GoosHarmsImbachSokolov2025} refuted this conjecture by proving that the $k$-Hamming distance problem, for $k = O(1)$, has constant sign rank; recall that it also  satisfies $\R(F) = O(1)$ and $\D^{\EQ}(F) = \Theta(\log n)$.

Prior to our work, the Sherstov--Storozhenko construction was the only potential candidate for a stronger separation: it satisfies $\R(F) = O(1)$, and by~\cite{GoosHarmsRiazanov2025}, $\D^{\EQ}(F) \gtrsim \sqrt{n}$; however, its sign rank is not known. Our example gives an optimal separation for this line of work, since point-line incidence problems have constant sign rank.

\paragraph{Statement regarding artificial
intelligence} A large language model had some part to play in the development of ideas leading to the disproof of the sum-product conjecture. (Its precise role is made clear in a statement towards the end of Section~1 in~\cite{BloomSawinSchildkrautZhelezov2026}.) While we have borrowed part of the resulting construction for our own proof, none of the novel mathematical ideas presented in this paper were generated by artificial intelligence; they are entirely the work of the two listed human authors.

\section{Number fields}
\label{sec:number_fields}

This section collects the facts about number fields used in our construction. We need only two things:~a quantitative sense in which nonzero algebraic integers cannot be small under all embeddings (Proposition~\ref{prop:norm}), and the existence of number fields in which the bounded set $B(K,C)$ is large (Lemma~\ref{lem:size} and Theorem~\ref{thm:martinet}).  

\paragraph{Number fields and embeddings}
A \emph{number field} $K$ is a field containing $\QQ$ that is finite-dimensional as a vector space over $\QQ$; this dimension $d$ is the \emph{degree} of $K$. An element of $K$ is an \emph{algebraic integer} if it is a root of some monic polynomial with integer coefficients, and the set $\O_K$ of algebraic integers forms a ring. The reader may find it useful to keep a concrete example in mind, such as the degree-$2$ number field
$$K = \QQ(\sqrt{5}) = \bigl\{p + q\sqrt 5 : p,q \in \QQ\bigr\},$$
whose ring of integers is
$$\O_K   = \biggl\{p + q\cdot \frac{1+\sqrt{5}}{2}: p,q \in \ZZ\biggr\}.$$
We will not require explicit descriptions such as these, only the fact that $\O_K$ is a ring.

A number field of degree $d$ admits exactly $d$ embeddings (injective field homomorphisms) $\sigma_1,\ldots,\sigma_d$ into $\CC$; each fixes $\QQ$ pointwise but acts differently on the rest of $K$. If every $\sigma_i$ maps $K$ into $\RR$, then $K$ is called \emph{totally real}. For instance, $\QQ(\sqrt 5)$ is totally real, with the two embeddings $p + q\sqrt 5 \mapsto p \pm q\sqrt 5$. Fixing one embedding lets us view a totally real $K$ as a subset of $\RR$, while the full collection $\sigma_1,\ldots,\sigma_d$ records the $d$ different ``views'' of each element as a real number.

\paragraph{Algebraic integers cannot be uniformly small}
The \emph{norm} $N(x)$ of $x \in K$ is the product $N(x) = \prod_{i=1}^d \sigma_i(x)$ of all its embeddings. The key property we use is the following.

\begin{proposition}[{\rm\cite{MR457396}}, Corollary 2 of Theorem 4\/]\label{prop:norm}
For any nonzero algebraic integer $x \in \O_K$, the norm $N(x)$ is a nonzero integer; in particular, $|N(x)| \ge 1$.
\end{proposition}

Proposition~\ref{prop:norm} implies that the $d$ views $\sigma_1(x),\ldots,\sigma_d(x)$ of a nonzero algebraic integer cannot all be close to $0$:~their product has absolute value at least $1$, so on average $|\sigma_i(x)| \ge 1$ in a geometric-mean sense. 

\paragraph{Number fields with many bounded integers}
For a real number $C$, consider the set of algebraic integers all of whose views are bounded by $C$:
\[
B(K, C) = \bigl\{ x \in \O_K : |\sigma_i(x)| \le C \text{ for all } 1 \le i \le d \bigr\}.
\]
The following size bound was established by Bloom, Sawin, Schildkraut, and Zhelezov in their disproof of the sum-product conjecture over the reals.

\begin{lemma}[{\rm\cite{BloomSawinSchildkrautZhelezov2026}}, Lemma~3.3\/]\label{lem:size}
Let $K$ be a totally real number field of degree $d$. For any real $C \ge 1$,
\[
\frac{C^d}{\sqrt{\Delta_K}} \;\le\; \bigl|B(K,C)\bigr| \;\le\; (2C+1)^d,
\]
where $\Delta_K$ is the \emph{discriminant} of $K$, a positive integer measuring the arithmetic complexity of $K$.
\end{lemma}

We will not define the discriminant here; for our purposes, it suffices that the disproof of the sum-product employed number fields of arbitrarily large degree $d$ with discriminant bounded by $O(1)^d$. The existence of such number fields is a result of Martinet~\cite{Martinet1978}.

\begin{theorem}[{\rm\cite{Martinet1978}}]\label{thm:martinet}
There is a constant $L > 0$ such that for infinitely many integers $d$, there exists a totally real number field $K$ of degree $d$ with discriminant $\Delta_K \le L^d$.
\end{theorem}

\section{The factorization norm}

We claimed in the introduction that finding a lower bound on equality-oracle complexity reduces to furnishing an arrangement of points and lines with many incidences. In this section we make this link clear by defining the \emph{$\gamma_2$ factorization norm}. For a real matrix $A$, this is defined to be
\begin{equation}
\normgamma{A} \;=\; \min_{BC = A} \normrow{B}\,\normcol{C},
\end{equation}
where the minimum is over all factorizations $A = BC$, $\normrow{B}$ is the maximum $\ell_2$-norm of a row of $B$, and $\normcol{C}$ is the maximum $\ell_2$-norm of a column of $C$. It is known (cf.~\cite{MR4494342,CheungHatamiHosseiniNikolovPitassiShirley2026}) that
\begin{equation}
\label{eq:complexity_factorization}
\R(F) \;\lesssim\; \widetilde\gamma_2(F)^2
\qquad\text{and}\qquad
\D^{\EQ}(F) \;\ge\; \frac{1}{2}\log\normgamma{F}.
\end{equation}

The \emph{degeneracy} of a Boolean matrix $F$ is the smallest integer $D$ such that every submatrix of $F$ has a row or a column with at most $D$ one-entries. Balla, Hambardzumyan, and Tomon proved the following sharp estimate for $\normgamma{F}$ in terms of degeneracy.

\begin{theorem}[{\rm\cite{bhtpublished26}}, Theorem~1.5\/]
\label{thm:degen}
Every Boolean matrix $F$ with degeneracy $D$ that contains no $2 \times 2$ all-ones submatrix satisfies
\[
\normgamma{F} \;=\; \Theta(\sqrt{D}).
\]
\end{theorem}

Since any two lines in $\RR^2$ cross at most once, every point-line incidence matrix $\PL_{U,V}$ is automatically free of $2 \times 2$ all-ones submatrices, so Theorem~\ref{thm:degen} applies directly. Combined with the Szemer\'edi--Trotter theorem~\cite{ST83}, this bounds the largest possible value of  $\normgamma{\PL_{U,V}}$.

\begin{corollary}
\label{cor:incidence_upper}
For any arrangement of $N$ lines and $N$ points in $\RR^2$,
\[
\normgamma{\PL_{U,V}} \;\lesssim\; N^{1/6}.
\]
\end{corollary}

\begin{proof}
By the Szemer\'edi--Trotter theorem, any arrangement of $m$ lines and $n$ points in $\RR^2$ has at most $O(m^{2/3}n^{2/3} + m + n)$ incidences. Applied to every submatrix, this bounds the degeneracy of $\PL_{U,V}$ by $O(N^{1/3})$, and the claim follows from Theorem~\ref{thm:degen}.
\end{proof}

In the other direction, a simple lower bound on degeneracy in terms of edge density turns incidences into a lower bound on the factorization norm.

\begin{corollary}
\label{cor:incidence_lower}
Let $U$ and $V$ be an arrangement of $N$ lines and $N$ points in $\RR^2$ with $k$ incidences. Then
\[
\normgamma{\PL_{U,V}} \;\gtrsim\; \sqrt{\frac kN}
\qquad\text{and}\qquad
\D^{\EQ}(\PL_{U,V}) \;\gtrsim\; \log\frac kN.
\]
\end{corollary}

\begin{proof}
Suppose that the number of $1$-entries in the matrix $\PL_{U,V}$ is $k$, and suppose that it has degeneracy $D$. Repeatedly removing a row or column in which the number of $1$-entries is at most $D$ empties the matrix in $2N$ steps. All of the $1$-entries are accounted for in this way, hence we see that $k < 2ND$. In other words, $D \gtrsim k/N$, and Theorem~\ref{thm:degen} supplies the first claim. The second claim then follows from~\eqref{eq:complexity_factorization}.
\end{proof}

Corollary~\ref{cor:incidence_lower} allows us to obtain $\D^{\EQ}(\PL_{U,V}) \gtrsim \log N$ by finding a point-line arrangement $(U,V)$ in which the lines of $U$ pass through $N^{\Omega(1)}$ points of $V$ on average.
 
\section{Proof of the main result}
\label{sec:main_proof}
Theorem~\ref{thm:main} follows immediately from the theorem below, by setting $\delta$ to any fixed constant smaller than $1/6$ and applying~\eqref{eq:complexity_factorization}. The theorem provides a very strong lower bound on the $\gamma_2$ factorization norm; to appreciate its strength, recall that the $\gamma_2$ factorization norm of an $N \times N$ Boolean matrix is always at most $\sqrt{N}$, via the trivial factorization $F = IF$, and that by Corollary~\ref{cor:incidence_upper}, any arrangement of $N$ lines and $N$ points $\PL_{U,V}$ always satisfies $\normgamma{\PL_{U,V}} \lesssim N^{1/6}$.

\begin{theorem}
\label{thm:main2}
For every $\delta > 0$, there exist sets $U, V \subseteq \RR^2$ of arbitrarily large size $|U| = |V| = N$ such that
\[
\R(\PL_{U,V}) = O_\delta(1)
\qquad\text{and}\qquad
\normgamma{\PL_{U,V}} \;\gtrsim\; N^{1/6 - \delta}.
\]
\end{theorem}

\begin{proof}
Let $L \ge 1$ be the constant from Theorem~\ref{thm:martinet}, and let $C = C(\delta, L) > 1$ be determined later. By Theorem~\ref{thm:martinet}, for infinitely many integers $d$ there exists a totally real number field $K$ of degree $d$ with discriminant $\Delta_K \le L^d$. Fix any such $K$, and set
\[
U = V = B\bigl(K, \sqrt{C}\bigr) \times B(K, C),
\]
with $N = |U| = |V|$.

\paragraph{Bounding the factorization norm} First we invoke Lemma~\ref{lem:size}, which gives the upper bound
\[
N \;\le\; (2\sqrt{C}+1)^d (2C+1)^d \;\le\; (9C)^{3d/2}
\]
on $N$. By Lemma~\ref{lem:size}, the number of $1$-entries in $\PL_{U,V}$ is at least
\[
\left(\frac{\sqrt{C/2}}{\sqrt{L}}\right)^{d} \left(\frac{\sqrt{C/2}}{\sqrt{L}}\right)^{d} \left(\frac{C/2}{\sqrt{L}}\right)^{d}
\;=\;
\left(\frac{C^2}{4L^{3/2}}\right)^{d},
\]
since for any $a, x \in B\bigl(K, \sqrt{C/2}\bigr)$ and any $1 \le i \le d$ we have $|\sigma_i(ax)| \le C/2$, so any $b \in B(K, C/2)$ gives us $|\sigma_i(ax + b)| \le C$; that is, each such triple $(a,x,b)$ corresponds to an incidence with $y=ax+b\in B(K,C)$. Choosing $C = C(\delta, L)$ large enough that
\[
\frac{C^2}{4L^{3/2}} \;\ge\; (9C)^{(3/2)\left(4/3- 2\delta\right)}
\]
guarantees
\[
\left(\frac{C^2}{4L^{3/2}}\right)^{d}\!\biggl/ N \;\gtrsim\; N^{1/3 - 2\delta}.
\]
Consequently, by Corollary~\ref{cor:incidence_lower},
\[
\normgamma{\PL_{U,V}} \;\gtrsim\; \sqrt{N^{1/3 - 2\delta}} \;=\; N^{1/6 - \delta},
\]
as desired.

\paragraph{Bounding the randomized communication complexity} It remains to prove the upper bound on
$\R(\PL_{U,V})$, which we record as a separate lemma.

\begin{lemma}\label{lem:r}
$\R(\PL_{U,V}) = O_\delta(1)$.
\end{lemma}

\begin{proof}
We exploit the fact that each embedding $\sigma_i$ is a field homomorphism, and hence preserves the algebraic structure of $K$. As before, suppose that Alice knows $a$ and $b$, Bob knows $x$ and $y$, and they want to know whether $z = ax+b-y = 0$.

Alice and Bob use their shared randomness to sample $i$ uniformly from $[d]$. Let $\eta > 0$ be chosen later. At communication cost $O\bigl(\log(C/\eta)\bigr)$, Alice sends real numbers $\alpha, \beta$ approximating $\sigma_i(a)$ and $\sigma_i(b)$ to within additive error $\eta$, i.e., $|\alpha - \sigma_i(a)| < \eta$ and $|\beta - \sigma_i(b)| < \eta$. Bob's estimate $\zeta = \alpha\sigma_i(x) + \beta - \sigma_i(y)$ of $\sigma_i(z)$ satisfies 
\[
|\zeta - \sigma_i(z)| = \bigl|(\alpha - \sigma_i(a))\sigma_i(x) + (\beta - \sigma_i(b))\bigr| \;\le\; (C+1)\eta.
\]

If $z = 0$, then $\sigma_i(z) = 0$ for every $i$. If $z \neq 0$, then by Proposition~\ref{prop:norm} and the \textmc{AM}--\textmc{GM} inequality,
\[
1 \;\le\; \bigl|N(z)\bigr|^{1/d} = \biggl(\prod_{i=1}^{d} |\sigma_i(z)|\biggr)^{1/d} \;\le\; \ex_{i \in [d]}\bigl[|\sigma_i(z)|\bigr].
\]
Splitting this expectation according to whether $|\sigma_i(z)| > 1/C$,
\[
1 \;\le\; C \cdot \Pr_{i \in [d]}\bigl[|\sigma_i(z)| > 1/C\bigr] + \frac{1}{C}\Bigl(1 - \Pr_{i \in [d]}\bigl[|\sigma_i(z)| > 1/C\bigr]\Bigr),
\]
which rearranges to
\[
\Pr_{i \in [d]}\bigl[|\sigma_i(z)| > 1/C\bigr] \;\ge\; \frac{1 - 1/C}{C - 1/C} \;=\; \frac{1}{C+1}.
\]

It remains to choose $\eta$ small enough that the estimate $\zeta$ lets Bob reliably distinguish $\sigma_i(z) = 0$ from $|\sigma_i(z)| > 1/C$. Taking
\[
\eta = \frac{1}{3C(C+1)},
\]
the additive error $|\zeta - \sigma_i(z)| \le (C+1)\eta = 1/(3C)$ is smaller than half of $1/C$, so Bob can correctly decide which case holds. Since $C$ is defined in terms of $\delta$ and the absolute constant $L$, and since $\eta$ depends only on $C$, the communication cost is $O\bigl(\log(C/\eta)\bigr) = O_\delta(1)$. The protocol is always correct when $z = 0$; when $z \neq 0$, it fails only if the sampled $i$ has $|\sigma_i(z)| \le 1/C$, which happens with probability at most $1 - 1/(C+1)$. Since the error is one-sided, repeating the protocol $O(C)$ times and reporting $z \neq 0$ if any repetition detects it reduces the error to $1/3$, increasing the communication by a factor of $O(C) = O_\delta(1)$.
\end{proof}
This lemma completes the proof of Theorem~\ref{thm:main2} (and consequently of Theorem~\ref{thm:main}).
\end{proof}
 
\section*{Acknowledgements}

Both authors are funded by the Natural Sciences and Engineering Research Council of Canada. We would
like to thank Hazem Hassan for helpful and enjoyable discussions, and are also grateful to Andrew Alexander, Nathan Harms, and Shachar Lovett for corrections on an earlier draft of this paper.

\bibliographystyle{alphacitation}
\begingroup\frenchspacing
\xpatchcmd{\em}{\itshape}{\slshape}{}{}
\bibliography{ref}
\endgroup\goodbreak

\bigskip\noindent
\textsc{Department of Mathematics and Statistics, McGill University, Montr\'eal, Qu\'ebec {\small H3A$\;$0B9}, Canada}

\smallskip\noindent
\textsl{E-mail address}: \texttt{marcel.goh@mail.mcgill.ca}

\bigskip\noindent
\textsc{School of Computer Science, McGill University, Montr\'eal, Qu\'ebec {\small H3A$\;$0E9
}, Canada}

\smallskip\noindent
\textsl{E-mail address}: \texttt{hamed.hatami@mcgill.ca}

\end{document}